\documentclass[12pt]{article}
\usepackage{makeidx}
\usepackage{amsfonts}
\usepackage{amsmath}
\usepackage{makeidx}
\usepackage{amssymb}
\usepackage{color}
\usepackage{enumerate}

\newtheorem{theorem}{Theorem}

\newtheorem{aexample}[theorem]{Example}

\newtheorem{lemma}[theorem]{Lemma}

\newtheorem{aproblem}[theorem]{Problem}

\newtheorem{acomment}[theorem]{Comment}

\newtheorem{aremark}[theorem]{Remark}
\newenvironment{remark}{\begin{aremark}\rm}{\end{aremark}}
\newtheorem{result}[theorem]{Result}
\newenvironment{proof}[1][Proof]{\noindent\textbf{#1.} }{\ \rule{0.5em}{0.5em}}
\numberwithin{equation}{section} \numberwithin{theorem}{section}

\begin{document}

\title{The absence of the selfaveraging property \\
of the entanglement entropy \\
of disordered free fermions \\
in one dimension}
\author{L. Pastur \\
B. Verkin Institute for Low Temperatures Physics \\
and Engineering, Kharkiv, Ukraine}
\date{}
\maketitle

\begin{abstract}
We consider the macroscopic system of free lattice fermions in one
dimensions assuming that the one-body Hamiltonian of the system is the one
dimensional discrete Schr\"odinger operator with independent identically
distributed random potential. We show that the variance of the entanglement
entropy of the segment $[-M,M]$ of the system is bounded away from zero as $%
M\rightarrow \infty $. This manifests the absence of the selfaveraging
property of the entanglement entropy in our model, meaning that in the
one-dimensional case the complete description of the entanglement entropy
is provided by its whole probability distribution.
This also may be contrasted the case of
dimension two or more, where the variance of the entanglement entropy per
unit surface area vanishes as $M\rightarrow \infty $ \cite{El-Co:17},
thereby guaranteing the representativity of its mean for large $M$ in
the multidimensional case.

PACS numbers 03.67.Mn, 03.67, 05.30.Fk, 72.15.Rn
\end{abstract}


\section{Problem and Results}

This note is an addition to the paper \cite{El-Co:17} by A. Elgart, M.
Shcherbina and the present author in which it is proved the following.
Consider the macroscopic system of free disordered fermions living on the $d$%
-dimensional lattice $\mathbb{Z}^{d}$ and having as the one-body Hamiltonian
the discrete Schr\"{o}dinger operator
\begin{equation}
H=-\Delta _{d}+V,  \label{SD}
\end{equation}%
where%
\begin{equation}
(\Delta _{d}u)(x)=-\sum_{|x-y|=1}u(y)+2d u(x),\;x\in \mathbb{Z}^{d}
\label{lad}
\end{equation}%
is the $d$-dimensional discrete Laplacian and
\begin{equation}
V=\{V(x)\}_{x\in \mathbb{Z}^{d}}  \label{pot}
\end{equation}%
is the random ergodic potential. Assume that the Fermi energy $E$ of the
system lies in the exponentially localized part of spectrum of $H$. This
means that the Fermi projection
\begin{equation}
P=\{P(x,y)\}_{x,y\in \mathbb{Z}^{d}}=\{\mathcal{E}_{H}(x,y;(-\infty
,E))\}_{x,y\in \mathbb{Z}^{d}}  \label{fp}
\end{equation}

of $H$, i.e., its spectral projection measure $\mathcal{E}_{H}((-\infty ,E))$
corresponding to the spectral interval \ $(-\infty ,E)$, admits the bound
\begin{equation}
\mathbf{E}\{|P(x,y)|\}\leq Ce^{-\gamma |x-y|},\;x,y\in \mathbb{Z}^{d}
\label{b_P}
\end{equation}%
for some $C<\infty $ and $\gamma >0$. Here and below the symbol $\mathbf{E}%
\{...\}$ denotes the expectation with respect to the random potential.

We refer the reader to \cite{El-Co:17} for the discussion of the cases where
the bound (\ref{b_P}) holds and guaranties the pure point spectrum of (\ref%
{SD}) with exponentially decaying eigenfunctions (exponential localization).
It is important for us in this paper that in the one-dimensional case $d=1$
the bound holds on the whole spectrum of $H$ if the potential (\ref{pot}) is
the collection of independent identically distributed (i.i.d.) random
variables.

Given the lattice cube (the block)
\begin{equation*}
\Lambda =[-M,M]^{d}\subset \mathbb{Z}^{d}
\end{equation*}%
of the system, define the entanglement entropy of free fermions as%
\begin{equation}
S_{\Lambda }=\mathrm{Tr}\ h(P_{\Lambda }),  \label{ee}
\end{equation}%
where%
\begin{equation}
h(t)=-t\log t-(1-t)\log (1-t),\;t\in \lbrack 0,1]  \label{h}
\end{equation}%
and
\begin{equation}
P_{\Lambda }=\{P(x,y)\}_{x,y\in \Lambda }  \label{fpl}
\end{equation}%
is the restriction of the Fermi projection \ref{fp}) to the block $\Lambda $.

It is proved in \cite{El-Co:17} that for any ergodic potential satisfying
condition (\ref{b_P}) of the exponential localization the entanglement
entropy satisfies the area law in the mean, i.e., there exists the limit%
\begin{equation}
\lim_{M\rightarrow \infty }\mathbf{E}\{L^{(d-1)}S_{\Lambda }\}=2d\,\mathbf{E}%
\{\mathrm{Tr}\,h(P_{\mathbb{Z}_{+}^{d}})\},\;L=2M+1,  \label{alm}
\end{equation}%
where $P_{\mathbb{Z}_{+}^{d}}=\{P(x,y)\}_{x,y\in \mathbb{Z}_{+}^{d}}$ is the
restriction of the Fermi projection (\ref{fp}) to the $d$-dimensional
lattice half-space%
\begin{equation*}
\mathbb{Z}_{+}^{d}=\mathbb{Z}_{+}\times \mathbb{Z}^{d-1},\;\mathbb{Z}%
_{+}=[0,1,..).
\end{equation*}%
See \cite{Ca-Co:11,Ei-Co:11,La:15,Le-Co:13} for various results on the
validity of the area law and its violation in translation invariant (non-random)
systems.

It was also shown in \cite{El-Co:17} that if the random potential is a
collection of i.i.d. random variables and (\ref{b_P}) holds, then there
exist some $C_{d}<\infty $ and $a_{d}>0$ such that
\begin{align}
\mathbf{Var}\{L^{(d-1)}S_{\Lambda }\}:= &\mathbf{E}\{(L^{(d-1)}S_{\Lambda
})^{2}\}-\mathbf{E}^{2}\{L^{(d-1)}S_{\Lambda }\}  \label{vare} \\
&\hspace{2cm}\leq C_{d}\,L^{-a_{d}},\;\;d\geq 2,  \notag
\end{align}%
i.e., that the fluctuations of the entanglement entropy per unit surface
area vanish as $L\rightarrow \infty $.

The relations (\ref{alm}) and (\ref{vare}) imply that in dimension two and
higher the entanglement entropy per unit surface area possesses the
selfaveraging property (see, e.g., \cite{Da-Co:14,Di-Co:98,LGP,Pa-Fi:92} for
discussion and use of the property in the condensed matter theory, spectral
theory and the quantum information theory where it is known as the
typicality).

On the other hand, it follows from the numerical results of \cite{Pa-Sl:14}
that for $d=1$ the fluctuations of the entanglement entropy of the lattice
segment $\Lambda =[-M,M]$ do not vanish as $M\rightarrow \dot{\infty}$ and
according to \cite{El-Co:17} we have for every typical realization (with
probability 1)%
\begin{equation}
S_{[-M,M]}=S_{+}(T^{M}\omega )+S_{-}(T^{-M}\omega )+o(1),\;M\rightarrow
\infty ,  \label{sp1}
\end{equation}%
where
\begin{equation}
S_{\pm }(\omega )=\mathrm{Tr}\,h(P_{\mathbb{Z}_{\mp }}(\omega)), \;\;  \mathbf{E}\{S_{\pm }\} >0. \label{spm}
\end{equation}%
Here and below $\omega =\{V(x)\}_{x\in
\mathbb{Z}}$ denotes a realization of random ergodic potential and $T$ is
the shift operator acting in the space of realizations of potential as $%
T\omega =\{V(x+1)\}_{x\in \mathbb{Z}}$. This suggests that for $d=1$ and
for i.i.d. potential the entanglement entropy of \ disordered free fermions
is not a selfaveraging quantity.

In this note we confirm the suggestion by establishing an $M$-independent
and strictly positive lower bound on the variance of the entanglement
entropy for $d=1$. Unfortunately, the class of random i.i.d. potentials, for
which this results is established, is somewhat limited (see, e.g. Remark \ref%
{r:f}). However, since the absence of selfaveraging property is not
completely common and sufficiently studied in the theory of disordered
systems, we believe that our result is of certain interest.

\begin{result}
Consider the macroscopic system of free lattice fermions in one dimension
whose one-body Hamiltonian is the discrete Schr\"{o}dinger operator (\ref{SD}%
) with i.i.d. potential (\ref{pot}). Assume that the common probability
distribution of $V(x),\;x\in \mathbb{Z}$ has a bounded density $f$ such that

\medskip (i) $\mathrm{supp\ }f=[0,\infty )$ and for some $\kappa >0$
\begin{equation}
\int_{0}^{\infty }v^{\kappa }f(v)dv<\infty ;  \label{cl2}
\end{equation}

\medskip (ii) the quantity
\begin{equation}
F(t)=J(t)-1,\;J(t)=\int_{0}^{\infty }\frac{f^{2}(v-t)}{f(v)}dv  \label{F}
\end{equation}%
is finite for all sufficiently large $t>0.$

Then there exist a sufficiently large $t_{0}>0$ and $M_{0}>0$ such that we
have for the entanglement entropy (\ref{ee}) uniformly in $M>M_{0}$
\begin{equation}
\mathbf{Var}\{S_{[-M,M]}\}:=\mathbf{E}\{S_{[-M,M]}^{2}\}-\mathbf{E}%
^{2}\{S_{[-M,M]}\}\geq A>0,  \label{var}
\end{equation}

\begin{equation}
A=\mathbf{E}^{2}\{S_{-}\}/F(t_{0})  \label{AF}
\end{equation}%
and $S_{-}$ is defined in (\ref{spm}).
\end{result}

\begin{remark}
\label{r:f} (i) It is easy to show (see (\ref{Fpos}) below) that $F\geq 0$.
Moreover, $F$ is unbounded as $t\rightarrow \infty $. Indeed, we have from (%
\ref{F}) and the Jensen inequality%
\begin{eqnarray*}
J(t) &=&\int_{0}^{\infty }\frac{f(v)}{f(v+t)}f(v)dv \\
&\geq &\left( \int_{0}^{\infty }\left( \frac{f(v)}{f(v+t)}\right)
^{-1}f(v)dv\right) ^{-1}=\int_{t}^{\infty }f(v)dv\rightarrow \infty
,\;t\rightarrow \infty .
\end{eqnarray*}%
Thus, $A$ in (\ref{var}) -- (\ref{AF}) can be rather small. Note that above
lower bound for $I$ is exact for the density $f(v)=a^{-1}e^{-av}$

(ii) Condition (i) of the result can be replaced by that for the support of $%
f$ to be bounded from below. However, a compact support is not allowed,
since in this case $J(t)$ in (\ref{F}) is not well defined for large $t$,
since the supports of the numerator and denominator in $J(t)$ do not
intersect. Moreover, even if the support of $f$ is the positive semiaxis, $f$
should not have zeros of the order 1 and higher.
\end{remark}

\textbf{Proof of result}. It follows from (\ref{alm}) for $d=1$ (or from (%
\ref{sp1})) that
\begin{eqnarray}  \label{fim}
\mathbf{E}\{S_{[-M,M]}\} &=&\mathbf{E}\{S_{+}(T^{M}\omega
)+S_{-}(T^{-M}\omega )\}+o(1) \\
&=&2\mathbf{E}\{S_{-}(\omega )\}+o(1),\;M\rightarrow \infty ,  \notag
\end{eqnarray}%
and in obtaining the second equality we used the shift and the reflection
invariance of the probability distribution of the infinite sequence $%
V=\{V(x)\}_{x\in \mathbb{Z}}$ of i.i.d. random variables, see \cite{El-Co:17}%
.

Likewise, repeating almost literally the proof of (\ref{alm}) for $d=1$ in
\cite{El-Co:17}, we obtain%
\begin{eqnarray}
\mathbf{E}\{S_{[-M,M]}^{2}\} &=&\mathbf{E}\{(S_{+}(T^{M}\omega
)+S_{-}(T^{-M}\omega ))^{2}\}+o(1)  \label{es2} \\
&=&2\mathbf{E}\{(S_{-}^(\omega ))^2\}+2\mathbf{E}\{S_{+}(T^{2M}\omega
)S_{-}(\omega )\}+o(1),\;M\rightarrow \infty .  \notag
\end{eqnarray}%
Since the infinite sequence $V=\{V(x)\}_{x\in \mathbb{Z}}$ of i.i.d. random
variables is a mixing stationary process (see e.g. \cite{Sh:95}, Section V.2
), we have%
\begin{eqnarray}  \label{ess}
\mathbf{E}\{S_{+}(T^{2M}\omega )S_{-}(\omega )\} &=&\mathbf{E}\{S_{+}(\omega
)\}\mathbf{E}\{S_{-}(\omega )\}+o(1)  \label{cs2} \\
&=&\mathbf{E}^{2}\{S_{-}(\omega )\}+o(1),\;M\rightarrow \infty .  \notag
\end{eqnarray}%
Combining (\ref{fim}) -- (\ref{ess}), we obtain%
\begin{equation}
\mathbf{Var}\{S_{[-M,M]}\}=2\mathbf{Var}\{S_{-}\}+o(1),\;M\rightarrow \infty
.  \label{v2v}
\end{equation}%
Thus, it suffices to show that $\mathbf{Var}\{S_{-}\}$ is strictly positive.

To this end we start with the inequality
\begin{equation*}
\mathbf{Var}\{\varphi (\xi ,\eta )\}\geq \mathbf{Var}\{\mathbf{E}\{\varphi
(\xi ,\eta )|\xi \}\}
\end{equation*}%
involving the conditional expectation and valid for any random
(multi-component in general) variables $\xi $ and $\eta $ and a function $%
\varphi $. Choosing here $\xi =V(0)$, $\eta =\{V(x)\}_{x\neq 0}$ and $%
\varphi=S_{-}$, we obtain%
\begin{equation*}
\mathbf{Var}\{S_{-}\}\geq \mathbf{Var}\{\mathbf{E}\{S_{-}|V(0)\}\}.
\end{equation*}%
Next, we will use Lemma \ref{l:cr} with $\xi =V(0)$ and $\varphi (\xi )=$ $%
\mathbf{E}\{S_{-}|V(0)=\xi \}$ yielding%
\begin{equation}
\mathbf{Var}\{S_{-}\}\geq \left. (\mathbf{E}\{S_{-}\}-\mathbf{E}%
\{S_{-}^{t}\})^{2}\right/ F(t),  \label{varl}
\end{equation}%
where $S_{-}^{t}$ is the entanglement entropy (\ref{ee}) -- (\ref{fpl})
corresponding to the Schr\"{o}dinger operator $H^{t}$ (see (\ref{SD}) -- (%
\ref{pot})) in which the potential $V(0)$ at the origin is replaced by $%
V(0)+t$%
\begin{equation}
H^{t}=H|_{V(0)\rightarrow V(0)+t}  \label{ht}
\end{equation}
We will prove below that%
\begin{equation}
\lim_{t\rightarrow \infty }\mathbf{E}\{S_{-}^{t}\}=0.  \label{leg}
\end{equation}%
Thus, there exists $\varepsilon(t) \to 0, \; t \to \infty $ (see, e.g (\ref%
{rem})) such that we have in view of (\ref{spm})%
\begin{equation*}
\mathbf{E}\{S_{-}\}-\mathbf{E}\{S_{-}^{t}\}) = (1 - \varepsilon (t))\mathbf{%
E}\{S_{-}\} >0, \;\; t \to \infty
\end{equation*}%
and then (\ref{v2v}) -- (\ref{varl}) yield (\ref{var}) -- (\ref{AF}) upon
choosing sufficiently large $t_0$ and $M_0$ and assuming that $M \ge M_0 $.

Let us prove (\ref{leg}). Since the potential is a collection of i.i.d.
random variables satisfying condition (\ref{cl2}), the spectrum of $H$ is
the positive semiaxis (see Corollary 4.23 in \cite{Pa-Fi:92}). The same is
true for the spectrum of $H^{t}$ since $t>0$. Hence, we have in view of (\ref%
{fp})
\begin{equation}
P=\mathcal{E}_{H}((0,E)),\quad P^{t}=\mathcal{E}_{H^{t}}((0,E)),\quad E>0.
\label{PE}
\end{equation}%
We have from the proof of Lemma 4.5 of \cite{El-Co:17} for some $t$%
-independent $C_0 < \infty$ and any $\alpha \in (0,1)$:%
\begin{align}
\mathbf{E}\{S_{-}^t\})& \leq C_{0}\sum_{x=0}^{\infty }\left( \sum_{y=-\infty
}^{-1}\mathbf{E}\{|P^{t}(x,y)|^{2}\}\right) ^{\alpha }  \label{espt} \\
& \leq C_{0}\mathbf{E}^{\alpha }\{P^{t}(0,0)\}+C_{0}\sum_{x=1}^{\infty
}\left( \sum_{y=-\infty }^{-1}\mathbf{E}\{|P^{t}(x,y)|^{2}\}\right) ^{\alpha
},  \notag
\end{align}%
where we took into account the inequality $(a+b)^{\alpha }\leq a^{\alpha
}+b^{\alpha },\;\alpha \in (0,1)$ and that $P^{t}$ is an orthogonal
projection, hence
\begin{equation*}
\sum_{y=-\infty }^{-1}|P^{t}(0,y)|^{2}\leq \sum_{y=-\infty }^{\infty
}|P^{t}(0,y)|^{2}=P^{t}(0,0).
\end{equation*}%
Now, (\ref{espt}) and Lemma \ref{l:pvan} below yield
\begin{equation}  \label{rem}
\mathbf{E}\{S^t_{-}\} \leq C_{\alpha,s}/ (t-E)^{\alpha s},
\end{equation}
where $C_{\alpha, s}$ does not depend on $t$. This implies (\ref{leg}) $%
\blacksquare$.

\section{Auxiliary results}

\begin{lemma}
\label{l:cr} Let $\xi $ be a non-negative random variable, $\varphi :\mathbb{%
R}_{+}\rightarrow \mathbb{R}$ be a function and $\mathbf{E}\{|\varphi (\xi
)|\}<\infty $. Assume that the probability law of $\xi $ has a bounded
density $f$ such that

(i) $\mathrm{supp}\,f=[0,\infty );$

(ii) the quantity $F(t)$ in (\ref{F}) 
is well defined for some $t>0$. Then we have%
\begin{equation}
\mathbf{Var}\{\varphi (\xi )\}\geq \left. (\mathbf{E}\{\varphi (\xi
)-\varphi (\xi+ t )\})^2\right/ F(t).  \label{cr}
\end{equation}%
\end{lemma}

\begin{proof}
Consider the random variables $\varphi (\xi )$ and $\psi (\xi )=(f(\xi
-t)-f(\xi ))/f(\xi )$. It follows from the normalization condition%
\begin{equation}  \label{nor}
\int_{0}^{\infty }f(x)dx=1
\end{equation}%
that
\begin{equation*}
\mathbf{E}\{\psi (\xi )\}=\int_{0}^{\infty }\left( \frac{f(x-t)-f(x)}{f(x)}%
\right) f(x)dx=0.
\end{equation*}%
Thus,%
\begin{eqnarray*}
\mathbf{Cov}\{\varphi (\xi )\psi (\xi )\} :=& \mathbf{E}\{(\varphi (\xi )-
\mathbf{E}\{\varphi (\xi )\})(\psi (\xi )-\mathbf{E}\{\varphi (\xi )\})\} \\
&\hspace{2cm}=\mathbf{E}\{\varphi (\xi )\psi (\xi )\}.
\end{eqnarray*}%
On the other hand, we have from the Schwarz inequality for the expectations:%
\begin{align*}
(\mathbf{Cov}\{\varphi (\xi )\psi (\xi )\})^{2}\leq \mathbf{Var}\{\varphi
(\xi )\}\mathbf{Var}\{\psi (\xi )\}.
\end{align*}%
Combining these two relations and using the definition of $\psi $, according
to which%
\begin{eqnarray}
\mathbf{Var}\{\psi (\xi )\} &=&\int_{0}^{\infty }\left( \frac{f(x-t)-f(x)}{%
f(v)}\right) ^{2}f(x)dx  \label{Fpos} \\
&=&\int_{0}^{\infty }\frac{f^{2}(x-t)}{f(x)}dx-1=F(t)\geq 0,  \notag
\end{eqnarray}%
we get (\ref{cr}).
\end{proof}

\begin{remark}
\label{r:cr} The inequality is, in fact, the Hammersley-Chapman-Robbins
inequality (see \cite{Le-Ca:98}, Section 2.5.1) and is a version of the
Cram\'er-Rao inequality of statistics.
\end{remark}

\begin{lemma}
\label{l:pvan} Let $H$ be the one-dimensional Schr\"{o}dinger operator (see (%
\ref{SD}) -- (\ref{pot}) for $d=1$) with an i.i.d. non-negative random
potential $V=\{V(x)\}_{x\in \mathbb{Z}}$ whose common probability law has  a
density $f$ satisfying (\ref{cl2}). Denote $P^{t}=\mathcal{E}%
_{H^{t}}((0,E))=\{P^{t}(x,y)\}_{x,y\in \mathbb{Z}}$ the Fermi projection of $%
H^{t}$ of (\ref{ht}) corresponding to the spectral interval $(0,E)$. Then
there exist $s\in (0,1),\;C<\infty $ and $\gamma >0$ that do not depend on $t
$ and are such that we have for $E<t$:

\medskip (i) \ $\mathbf{E}\{|P^{t}(0,0)|\}\leq C(t-E)^{-s}$;

\medskip (ii) $\mathbf{E}\{|P^{t}(x,y)|\}\leq C(t-E)^{-s}e^{-\gamma (x-y)},$
for all $x\geq 1$ and $y\leq -1$.
\end{lemma}

\begin{proof}
Let%
\begin{equation*}
G^{t}(z)=(H^{t}-z)^{-1}=\{G^{t}(x,y;z)\}_{x,y\in \mathbb{Z}},\;z\in \mathbb{C%
}\setminus \mathbb{R}
\end{equation*}%
be the resolvent of $H^{t}$. It is shown below that the bounds%
\begin{equation}
\mathbf{E}\{|G^{t}(0,0;\lambda +i\eta )|^{s}\}\leq C/(t-E)^{s}  \label{eg00}
\end{equation}%
and%
\begin{equation}
\mathbf{E}\{|G^{t}(x,y;\lambda +i\eta )|^{s}\}\leq C/(t-E)^{s}e^{-\gamma
(x-y)}  \label{egxy}
\end{equation}%
are valid for some $s\in (0,1)$, all $\lambda \in (0,E),\;E<t,\;\eta \neq 0$
and $x\geq 1$, $y\leq -1$ with $C<\infty $ and $\gamma >0$ which are independent of $%
t,\;\lambda  $ and $\eta \neq 0$.

It follows from a slightly modified version of proof of Theorem 13.6 of \cite%
{Ai-Wa:15}, based on the contour integral representation of $P^t$ via $G^t$
and the Combes-Thomas theorem, that the assertion of the lemma can be
deduced from (\ref{eg00}) -- (\ref{egxy}).

Hence, it suffices to prove (\ref{eg00}) and (\ref{egxy}). To this end we
introduce the restrictions $H_{-}$ and $H_{+}$ of $H^{t}$ (or $H$) to the
integer-valued intervals $(-\infty ,-1]$ and $[1,\infty )$ and the rank one
operator $\widehat{V}_{0}^{t}$ of multiplication by $V(0)+2+t$. Let%
\begin{equation}
\widehat{H}=H_{-}\oplus \widehat{V}_{0}^{t}\oplus H_{+}  \label{hh}
\end{equation}%
be the double infinite block matrix consisting of the $(-\infty ,-1]\times
(-\infty ,-1]$ semi-infinite block $H_{-}$, $1\times 1$ "central" block $%
\widehat{V}_{0}^{t}$ and the $[1,\infty )\times \lbrack 1,\infty )$
semi-infinite block $H_{+}$. In other words, $\widehat{H}^{t}$ is obtained
from $H^{t}$ by replacing the four entries (equal $-1)$ with indices $%
(0,\pm 1)$ and $(\pm 1,0)$ by zero. Denote
\begin{align}
\widehat{G}^{t}(z)& =(\widehat{H}^{t}-z)^{-1}=\{\widehat{G}%
^{t}(x,y)\}_{x,y\in \mathbb{Z}},  \label{gpm} \\
G_{\pm }(z)& =(H_{\pm }-z)^{-1}=\{G_{\pm }(x,y)\}_{x,y\in \lbrack \pm 1,\pm
\infty )}  \notag
\end{align}%
the corresponding resolvents, where we omit the complex spectral parameters $%
z,\;\Im z\neq 0$ in the r.h.s. We have in view of (\ref{hh}):%
\begin{equation}
\widehat{G}^{t}=G_{-}\oplus (\widehat{V}_{0}^{t}-z)^{-1}\oplus G_{+}.
\label{ght1}
\end{equation}%
By using the resolvent identity $G^{t}=\widehat{G}^{t}-\widehat{G}^{t}(H^{t}-%
\widehat{H}^{t})G^{t}$, we obtain for all $x,y\in \mathbb{Z}$%
\begin{eqnarray*}
G^{t}(x,y) &=&\widehat{G}^{t}(x,y)+\widehat{G}%
^{t}(x,0)(G^{t}(-1,y)+G^{t}(1,y)) \\
&&+(\widehat{G}^{t}(x,1)+\widehat{G}^{t}(x,1))G^{t}(0,y).
\end{eqnarray*}%
This and (\ref{ght1}) imply%
\begin{equation}
G^{t}(0,y)=(V(0)+2+t-z)^{-1}(\delta _{0,y}+G^{t}(-1,y)+G^{t}(1,y)),\;y\leq 0
\label{gt0y}
\end{equation}%
and
\begin{equation}
G^{t}(x,y)=G^{t}(0,y)G_{+}(x,1),\;x\geq 1,\;y\leq -1.  \label{gt1}
\end{equation}%
Likewise, we have from the resolvent identity $G^{t}=\widehat{G}%
^{t}-G^{t}(H^{t}-\widehat{H}^{t})\widehat{G}^{t}$:%
\begin{equation}
G^{t}(x,0)=(V(0)+2+t-z)^{-1}(\delta _{x,0}+G^{t}(x,-1)+G^{t}(x,1)),\;x\geq 0
\label{gtx0}
\end{equation}%
and%
\begin{equation}
G^{t}(x,y)=G^{t}(x,0)G_{-}(-1,y),\;x\geq 1,y\leq -1.  \label{gt2}
\end{equation}%
We have then from (\ref{gt0y}) and the Schwarz inequality for any $s>0$ and
all $y\leq 0$
\begin{eqnarray}
\mathbf{E}\{|G^{t}(0,y)|^{s}\} &\leq &C_{s}\,(g(z-2-t))^{1/2}  \label{egt0y} \\
&&(\delta _{0,y}+\mathbf{E}^{1/2}\{|G^{t}(-1,y)|^{2s}\}+\mathbf{E}%
^{1/2}\{|G^{t}(1,y)|^{2s}\}),  \notag
\end{eqnarray}%
where $C_{s}$ depends only on $s>0$ and
\begin{equation*}
g(\zeta )=\mathbf{E}\{|V(0)-\zeta |^{-2s}\}`=\int_0^\infty \frac{f(v)dv}{%
|v-\zeta |^{2s}}.
\end{equation*}%
Choosing here $\zeta =z-2-t,\;z=\lambda +i\eta $ and using (\ref{nor}), we
have for $\lambda \leq E<t$
\begin{equation}
g(z-2-t)\leq (t-E)^{-2s}  \label{gb}
\end{equation}%
and then (\ref{egt0y}) implies for any $s>0$ and all $y\leq 0$
\begin{align}
& \mathbf{E}\{|G^{t}(0,y;z)|^{s}\}\leq C_{s}(t-E)^{-s}  \label{gt0ys} \\
& \hspace{0.5cm}\times (\delta _{0,y}+\mathbf{E}^{1/2}\{|G^{t}(-1,y;z)|^{2s}%
\}+\mathbf{E}^{1/2}\{|G^{t}(1,y;z)|^{2s}\}).  \notag
\end{align}%
We will use now Theorem 8.7 of \cite{Ai-Wa:15}, according to which if $A_{0}$
is a selfadjointe operator in $l^{2}(\mathbb{Z}^{d})$, $U=\{U(x)\}_{x\in
\mathbb{Z}^{d}}$ is a collection of independent random variables whose probability
densities $f_x,\, x \in \mathbb{Z}$ are bounded uniformly in $x$, i.e.,
$\sup_{x \in \mathbb{Z}} \sup_{v \in \mathbb{R}} < \infty$ and if $\mathcal{G}(z)=(A-z)^{-1}%
=\{\mathcal{G}
(x,y;z)\}_{x,y\in \mathbb{Z}^{d}}$ is the resolvent of $A=A_{0}+U$, then for
any $s\in (0,1)$ there exists $C_{s}^{\prime }<\infty $ such that the bound%
\begin{equation*}
\mathbf{E}\{|\mathcal{G}(x,y;z)|^{s}\}\leq C_{s}^{\prime }
\end{equation*}%
holds uniformly in $z\in \mathbb{C\setminus R}$\ for all $x,y\in \mathbb{Z}%
^{d}$.

Choosing here $A=H^{t}$ with   $H^{t}$ of (\pageref{ht}) and noting that
for the potential $V^t$ the conditions of the theorem are satisfied (all
$V^t(x)=V(x), \, x\neq 0$ are i.i.d random variables with a bounded common probability
density and $V^t(0)=V(0)+t$ has the density $f(v-t)$ also bounded), we obtain for any $s\in (0,1),\;\lambda <E, \; \eta
\neq 0$ and all $x,y\in \mathbb{Z}^{d}$
\begin{equation}
\mathbf{E}\{|G^{t}(x,y;\lambda +i\eta )|^{s}\}\leq C_{s}^{\prime \prime },
\label{gtxyb}
\end{equation}%
where $C_{s}^{\prime \prime }$ does not depend on $t,E$ and $\eta $.

Plugging this bound with $x=\pm 1$ into (\ref{gt0ys}), we get for any $%
s_{1}\in (0,1/2),\;\lambda <E<t$ and all $y\leq 0$
\begin{equation}
\mathbf{E}\{|G^{t}(0,y;\lambda +i\eta )|^{s_{1}}\}\leq
B_{s_{1}}/(t-E)^{s_{1}},  \label{eg0yf}
\end{equation}%
where $B_{s_{1}}$ does not depend on $t,E$ and $\eta $.

Analogous argument yields for $s_{1}\in (0,1/2),\;\lambda <E<t$ and all $%
x\geq 0$
\begin{equation}
\mathbf{E}\{|G^{t}(x,0;\lambda +i\eta )|^{s_{1}}\}\leq
B_{s_{1}}/(t-E)^{s_{1}}.  \label{egx0f}
\end{equation}%
We obtain (\ref{eg00}), hence assertion (i) of the lemma, from (\ref{eg0yf})
with $y=0$ (or from (\ref{egx0f}) with $x=0$).

To prove (\ref{egxy}), hence assertion (ii) of the lemma, we combine (\ref%
{gt1}) and (\ref{gt2}) to write for any $s>0$ and all $x\geq 1$ and $y\leq 1$%
:
\begin{align*}
|G^{t}(x,y;z)|^{s}& =|G^{t}(x,0;z)|^{s/2}|G^{t}(0,y;z)|^{s/2} \\
\times| & G_{+}(x,1;z)|^{s/2}|G_{-}(-1,y;z)|^{s/2}
\end{align*}%
and then, by H\"{o}lder inequality for expectations,%
\begin{eqnarray}
&&\mathbf{E}\{|G^{t}(x,y;z)|^{s}\}\leq \mathbf{E}^{1/4}\{|G^{t}(x,0;z)|^{2s}%
\}\mathbf{E}^{1/4}\{|G^{t}(0,y)|^{3s}\}  \label{g4} \\
&&\hspace{5cm}\times \mathbf{E}^{1/4}\{|G_{+}(x,1)|^{2s}\}\mathbf{E}%
^{1/4}\{|G_{-}(-1,y)|^{2s}\}.  \notag
\end{eqnarray}%
Using here (\ref{eg0yf}) and (\ref{egx0f}), we get for $s_{1}\in (0,1/2)$,
and all $x\geq 1$ and$\;y\leq -1$

\begin{equation}
\mathbf{E}\{|G^{t}(x,y;z)|^{s_{1}}\}\leq \frac{B_{s_{1}}}{(t-E)^{s_{1}}}%
\mathbf{E}^{1/4}\{|G_{+}(x,1)|^{2s_{1}}\}\mathbf{E}^{1/4}%
\{|G_{-}(-1,y)|^{2s_{1}}\}.  \label{gsi}
\end{equation}%
To bound the two last factors on the right, we will use a result from \cite%
{Mi:96} according to which if $H_{+}$ is the discrete one dimensional Schr%
\"{o}dinger operator in $l^{2}([1,\infty ))$ with i.i.d. potential whose
common probability law is such that $\mathbf{E}\{|V(0)|^{\kappa }\}<\infty $
for some $\kappa >0$, then 
 for any spectral interval $I$ there exist $C(I)<\infty ,\,\ \gamma (I)>0
$ and $s_{2}\leq \kappa /2$ such that%
\begin{equation}
\mathbf{E}\{|G_{+}(x,1)|^{s_{_{2}}}\}\leq C(I)e^{-\gamma (I)x},\;x\geq 1.
\label{min1}
\end{equation}%
The same is valid for the Hamiltonian $H_{-}$ acting in $l^{2}((-\infty ,-1])
$, see (\ref{hh}):
\begin{equation}
\mathbf{E}\{|G_{+}(-1,y)|^{s_{_{2}}}\}\leq C(I)e^{\gamma (I)y},\;y\leq -1.
\label{min2}
\end{equation}%
The bounds are the basic ingredient of the proof of (\ref{b_P}) for the one
dimensional case \cite{Mi:96}.

Using these bounds in the r.h.s. of (\ref{gsi}), we obtain assertion (ii) of
the lemma with $s=\min \{s_{1},s_{2}\}$, where $s_{1}$ and $s_{2}$ are
defined in (\ref{eg0yf}) -- (\ref{egx0f}) and (\ref{min1}) -- (\ref{min2}).
\end{proof}

\begin{remark}\label{r:weyl}
By using the standard facts of spectral theory, it is easy to prove the
weaker version of Lemma \ref{l:pvan}%
\begin{equation}
\lim_{t\rightarrow \infty }P^{t}(x,y)=0, \; x\geq 0,\,y\leq 0  \label{lp}
\end{equation}%
which, however, does not allow us to justifies the limiting transition $t\rightarrow
\infty $ in the second term in the r.h.s. of (\ref{espt}).

Indeed, according to the spectral theorem%
\begin{equation*}
G^{t}(x,y;z)=\int_{-\infty }^{\infty }\frac{\mathcal{E}_{H^{t}}(x,y;d%
\lambda )}{\lambda -z},\,\, \Im z\neq 0.
\end{equation*}%
The formula, the continuity properties of the Stieltjes transform of a
bounded signed measure and (\ref{fp}) imply that (\ref{lp}) follows from the
analogous limiting relation for the resolvent with $\Im z\neq 0$:%
\begin{equation}
\lim_{t\rightarrow \infty }G^{t}(x,y;z)=0,\,\, x\geq 0,\,y\leq 0.  \label{lg}
\end{equation}%
Viewing the part $(t-1)V(0)$ of the $(0,0)$th  entry of $H^{t}$ as a rank
one perturbation of $H=H^{t}|_{t=1}$, we obtain%
\begin{equation*}
G^{t}(x,y;z)=G(x,y;z)-\frac{(t-1)V(0)G(x,0;z)G(0,y;z)}{1+(t-1)V(0)G(0,0;z)},
\end{equation*}%
hence%
\begin{equation}
\lim_{t\rightarrow \infty }G^{t}(x,y;z)=G(x,y;z)-\frac{G(x,0;z)G(0,y;z)}{%
G(0,0;z)}.  \label{lgt}
\end{equation}%
Recall now the Weyl formula for the resolvent of the discrete
Schr\"odinger operator (see, e.g. \cite{Te:99}, Section 1.2):%
\begin{equation*}
G(x,y;z)=G(0,0;z)\left\{
\begin{array}{cc}
\psi _{+}(x;z)\psi _{-}(y;z), & x\geq y, \\
\psi _{-}(x;z)\psi _{+}(y;z), & x\leq y,%
\end{array}%
\right.
\end{equation*}%
where $\psi ^{\pm}$ are the  solutions of the corresponding discrete
Schr\"odinger equation which belong to $l^{2}(\mathbb{Z}_{\pm })$ for
$\Im z \neq 0 $ and satisfy
the condition $\psi ^{\pm }(0)=1$. Combining the formula with (\ref{lgt}),
we obtain (\ref{lg}), hence (\ref{lp}).
\end{remark}

\textbf{Acknowledgment} The work is supported in part by the grant 4/16-M of
the National Academy of Sciences of Ukraine.

\end{document}